\documentclass{llncs}
\usepackage{amsmath}
\usepackage{amssymb}
\usepackage{fancybox}
\usepackage{graphicx}
\newcommand{\COMMENTED}[1]{{}}
\newcommand{\junk}[1]{\COMMENTED{#1}}
\begin{document}

\title{A fast parallel algorithm for minimum-cost
 small integral flows}
%\title{Integral network flow parametrized
%by flow value admits an efficient
%paralllelization}
\author{
%Leszek G\c asieniec \inst{1}
%\and
Andrzej Lingas \inst{1}
\and
Mia Persson \inst{2} 
\institute{
%Department of Computer Science, 
%University of Liverpool, Peach Street, L69 7ZF, UK. \texttt{leszek@csc.liv.ac.uk.}
%\and 
Department of Computer Science, Lund University, 22100
Lund, Sweden. \texttt{Andrzej.Lingas@cs.lth.se. Fax +46 46 13 10 21}
\and
Department of Computer Science, Malm\"o University, 205 06 Malm\"o, Sweden. \texttt{mia.persson@mah.se}
%\and
%The Centre for Mathematical Sciences, Lund University, 22100 Lund, Sweden.
%\texttt{Dzmitry.Sledneu@math.lu.se}
}}
 \date{} 
\maketitle
\begin{abstract}
We present a new approach to the minimum-cost 
integral flow problem for small values of the flow.
It reduces the problem to the tests of simple multi-variable
polynomials over a finite field of characteristic two
for non-identity with zero. In effect, we show that
a minimum-cost flow of value $k$ in a network
with $n$ vertices, a sink and a source, 
integral edge capacities and positive integral
edge costs polynomially bounded in $n$
can be found by a randomized PRAM, with errors
of exponentially small probability in $n,$ running 
in $O(k\log (kn)+\log^2 (kn))$ time and using
$2^{k}(kn)^{O(1)}$ processors. Thus, in particular,
for the minimum-cost flow of value $O(\log n),$ we obtain 
an $RNC^2$ algorithm.
\junk{We consider a parametrization of
the minimum-cost flow problem
with respect to the flow value. The minimum-cost flow problem
is a generalization
of the maximum flow problem which is known to be $P$-complete. 
We show that the minimum cost of
a flow of value $k$ in a network
with $n$ vertices, a sink and a source, 
integral edge capacities and positive integral
edge costs polynomially bounded in $n$
can be found by a randomized PRAM , with
one-sided errors of exponentially small in $n$
probability, running in $k\log^{O(1)}(kn)$ time
and using $2^{k}k^{O(1)}n^{O(1)}$ processors.
Thus, we establish the membership of so
parametrized minimum-cost integral network problem
in the randomized version of the class of {\em fixed-parameter 
parallelizable problems} (FPP). 
We also show that
the problem of finding a minimum-cost flow of value $k$ in a network
with a sink and a source, integral edge capacities
and positive polynomially bounded integral
edge costs 
admits a randomized FPP algorithm. 
Finally, we 
conclude that the problem of finding a
minimum-cost integral flow
of value (at most) logarithmic in the network size admits
an RNC algorithm.}
\end{abstract}

%\end{titlepage}

\section{Introduction}
The maximum network flow problem is a well known
fundamental problem in algorithms and optimization
with plenty of important applications \cite{AMO93,E79,EK72,L76}. It is known
to be $P$-complete even in its integral version 
provided that the edge capacities are
exponentially large in the size
of the network \cite{GSS82}.
The minimum-cost flow problem is a well known important
generalization of the maximum flow problem \cite{AMO93,EK72,F61,L76}. 
The objective
is to compute a maximum flow of minimum cost in a directed
graph where each edge is assigned a cost. For a flow $f$
in a directed graph $(V,E),$ the cost of $f$ is simply
$\sum_{e\in E}f(e)cost(e).$

The prospects for designing a fast and
processor efficient parallel algorithm, in particular
an NC algorithm \cite{R93}, for maximum integral flow 
or minimum-cost integral flow are
small. The fastest known parallel implementations
of general maximum flow and/or  minimum-cost flow algorithms achieve solely
a moderate speed up and still run in 
$\Omega(n^{\alpha})$ time, where $\alpha$ is a positive
constant, see \cite{AS92,G93}.

The situation changes when the edge capacities or the supply of
flow as well as edge costs are substantially bounded. For example, if the edge
capacities and edge costs are bounded by a polynomials in $n,$ both
problems admit RNC algorithms. Then, the maximum integer flow
problem admits even an $RNC^2$ algorithm \cite{MVV87,OS92,V93} while 
the minimum-cost integer flow problem admits an $RNC^3$
algorithm \cite{OS92}. At the heart of the aforementioned
RNC solutions is the randomized method of detecting a perfect
matching by randomly testing Edmonds' multi-variable polynomials
for non-identity with zero \cite{E67,KUW86,MVV87,V93}.

When the flow supply is relatively small,
e.g.,  logarithmic in the size
of the network or a poly-logarithmic one,
then just an NC implementation of the basic phase
in the standard Ford-Fulkerson method \cite{AMO93,E79,EK72,FF56,L76} yields
an NC algorithm ($NC^3$ when the supply is logarithmic)
for maximum integer flow that can be extended to
an NC algorithm for minimum-cost integer flow (when
edge costs are polynomially bounded). 
The number of processors used corresponds to that
required by a shortest path computation.
%, using a number of processors
%corresponding to the time complexity of matrix multiplication.
% that can be extended to
%an NC algorithm for minimum-cost integer flow (when
%edge costs are polynomially bounded). 
%The NC
%algorithm for maximum flow uses 
%a number of processors corresponding to the
%time complexity of matrix multiplication.

In this paper,
we present a new approach to the minimum-cost 
integral flow problem for a small value $k$ of the flow.
We directly associate a simple polynomial over a finite field with
the corresponding problem 
of the existence of $k$ mutually vertex disjoint paths
of bounded total length,
connecting two sets of $k$ terminals in a directed graph.
By using the idea of monomial cancellation, the latter
problem reduces to testing the polynomial over a finite field
of characteristic two for non-identity
with zero. 
We combine the  DeMillo-Lipton-Schwartz-Zippel
lemma \cite{DL78,S80} on probabilistic
verification of polynomial identities
with parallel dynamic programming
to perform the test efficiently in parallel.
Additionally, we use the isolation lemma 
to construct the minimum-cost flow \cite{MVV87,V93}.
\junk{It reduces the problem to the tests of simple multi-variable
polynomials over a finite field of characteristic two
for non-identity with zero.} 

In effect, we infer that
a minimum-cost flow of value $k$ in a network
with $n$ vertices, a sink and a source, 
integral edge capacities and positive integral
edge costs polynomially bounded in $n$
can be found by a randomized PRAM, with errors
of exponentially small probability in $n$, 
running 
in $O(k\log (kn)+\log^2 (kn))$ time and using
$2^{k}(kn)^{O(1)}$ processors. Thus, in particular,
for the minimum-cost flow of value $O(\log n),$ we obtain 
an $RNC^2$ algorithm.

\junk{
For these reasons, it is natural to look
at parametrized versions of these problems
and if possible to design efficient parallel
solutions to the slices of the problems
corresponding to different fixed values of
the parameter. In case of the integral
flow problems, the integral value of the flow 
is a natural parameter.

For the maximum integral flow problem,
we can easily obtain a fast parallel algorithm
for the $k$-th slice, i.e., for finding
a flow of value $k$ just by parallelizing
the finding of an augmenting path in the
residual graph in the Ford-Fulkerson algorithm
and iterating it up to $k$ times \cite{}.
Unfortunately, the case of 
the parametrized minimum-cost integral flow problem
seems much harder.}
\par
\vskip 3pt
\noindent
{\bf Related work.}
For the RNC algorithms for the related problem
of minimum-cost perfect matching see \cite{KUW86,MVV87,V93}.
For the comparison of time and substantial processor complexities
of prior RNC algorithms for the minimum-cost flow see 
page 7 in \cite{OS92}.
The fastest of the reported algorithms is not
an $RNC^2$ one even when the flow supply and thus the edge capacities
are logarithmic in the size of the network.
The idea of associating a polynomial over
a finite field to the sought structure has been already
used by Edmonds to detect matching \cite{E67} and then in several
papers presenting
RNC algorithms for perfect matching construction \cite{KUW86,MVV87,V93}. 
It appears in several recent papers that also exploit the idea
of monomial cancellation \cite{B10,BHT12,K08,W09}.

\junk{One could utilize the maximum flow minimum cut theorem \cite{E79}
to solve the $k$-th slice problem, i.e., to answer the question
if there is a flow of value $k$ from
the source $s$ to the sink $t,$ by
verifying if the network has an $s-t$ edge cut
of capacity not exceeding $k-1$ . In order 
to perform the verification fast in parallel
one could test all subsets of edges of total capacity
not exceeding $k-1$ for being an $s-t$ cut.
Note however that the total number of such subsets,
and consequently the work performed by such an
algorithm, could be $n^{\Omega(k)}.$ 

We would like to solve the $k$-th slice problem
of the parametrized minimum-cost integral flow problem
in $f(k)\log ^{O(1)} n$ time
using $g(k)n^{O(1)}$ processors, where $n$
is the number of vertices in the network and
$f,\ g$ are natural functions. 
%Indeed, we show that this is possible if
%random bits ca be used. Namely, we ....

%In fact, Bodlaender, Downey and Fellows \cite{BDF94},
In fact, Cesati and Di Ianni \cite{CD97} extended the framework 
of parametrized complexity \cite{DF96} by introducing
two classes of efficiently parallelizable parametrized
problems, PNC and FPP, respectively\footnote
{They refer to a private communication 
from Bodlaender, Downey and Fellows from 1994
as the origin of the concept of the class PNC.}.
PNC (parametrized analog of NC) contains all parametrized
problems which have a parallel deterministic algorithm
(e.g., PRAM \cite{R93})
running in $f(k)(\log |x|)^{h(k)}$ time and using $g(k)|x|^{\beta}$
processors, where $<x,k>$ is the instance of the problem.
$k$ is the parameter, $f,g$ and $h$ are arbitrary functions,
and $\beta $ is a constant. The
class of {\em fixed-parameter parallelizable problems}, for short
FPP, is a subclass of PNC, where
the function $h$ is required to be a constant.

One can naturally consider randomized counterparts of
these two classes (in analogy to the
random NC class, RNC \cite{R93}) by allowing the parallel algorithm to
use random bits and return an answer with a one-side bounded error.
\par
\vskip 3pt
\noindent
{\bf Our results.}
In terms of
the aforementioned terminology, 
we show that the parametrized minimum-cost integral flow
problem is in the randomized version of the class FPP. 
Namely, we show that the minimum cost of
a flow of value $k$ in a network
with $n$ vertices, a sink and a source, 
integral edge capacities and positive integral
edge costs polynomially bounded in $n$
can be found by a randomized PRAM, with
errors of exponentially small probability in $n,$
%one-sided errors of exponentially small in $n$
%probability, 
running in $k\log^{O(1)}(kn)$ time
and using $2^{k}k^{O(1)}n^{O(1)}$ processors.
We also show that
the problem of finding a minimum-cost flow of value $k$ in a network
with a sink and a source, integral edge capacities
and positive polynomially bounded integral
edge costs (if possible)
admits a randomized FPP algorithm. 
Finally, we 
conclude that the problem of finding a
minimum-cost integral flow
of value (at most) logarithmic in the network size admits
an RNC algorithm.
\par
\vskip 3pt
\noindent
{\bf Techniques.}
We associate a polynomial over a finite field with
the corresponding problem 
of the existence of $k$ mutually vertex disjoint paths
of bounded total length,
connecting two sets of $k$ terminals in a directed graph.
By using the idea of monomial cancellation, the latter
problem reduces to testing the polynomial for non-identity
with zero. We combine the  DeMillo-Lipton-Schwartz-Zippel
lemma \cite{DL78,S80} on probabilistic
verification of polynomial identities
with parallel dynamic programming
to perform the test efficiently in parallel.
%As for the finding variant, we
%use additionally the isolation lemma \cite{V93}.
Note that the idea of associating a polynomial over
a finite field to the sought structure has been already
used by Edmonds to detect matching \cite{E67}, and it appears
in several recent papers that also exploit the idea
of monomial cancellation \cite{B10,BHT12,K08,W09}.}
\par
\vskip 3pt
\noindent
{\bf Organization.} In the next section, we comment
briefly on the basic notation and the model of parallel computation
used in the paper. In Section 3, we derive our fast
randomized parallel method for detecting the existence of
$k$ mutually vertex disjoint paths of bounded total length
connecting two sets of $k$ terminals in a directed graph.
In Section 4, we generalize the method
to include edge costs which enables us to replace
the total length bound with the total cost one.
%In Section 4, we generalize the method
%to include finding of the sought $k$ paths.
In section 5, we show a straightforward reduction
of the minimum-cost integer flow problem parametrized
by the flow value to the
corresponding disjoint paths problem which enables
us to derive our main result on detecting minimum-cost
small flows in parallel.

\section{Terminology}

For a natural number $n,$ we let $[n]$ denote
the set of natural numbers in the interval $[1,n].$
The cardinality of a set $A$ will be denoted by $|A|.$

We assume the standard definitions of {\em flow}
and {\em flow value}
in a network (directed graph) with 
integral edge capacities, a distinguished
source vertex $s$ and a distinguished sink vertex $t$
(e.g., see \cite{E79}) .

For the definitions of parallel random access machines (PRAM),
the classes NC and RNC and the corresponding notions
of NC and RNC algorithms, the reader is referred to \cite{R93}.

The {\em characteristic} of a ring or a field is the
minimum number of $1$ in a sum that yields $0.$
A finite field with $q$ elements is often denoted by $F_q.$

\junk{
By an FPP algorithm for a problem parametrized
with respect to $k,$ we mean 
a PRAM running in $f(k)\log^{O(1)}n$ time
and using $g(k)n^{O(1)}$ processors, where
$f,g$ are natural functions.}

\section{Connecting vertex-disjoint paths}

It is well known that the maximum integral network flow
problem with bounded edge capacities corresponds
to a disjoint path problem
(cf. \cite{E79}). In Section 5, we provide  an efficient
parallel reduction of the minimum-cost integral flow
problem parametrized by the flow value
to a  parametrized disjoint path problem. 
This section is devoted to a derivation of a fast randomized
parallel method for the decision version of the parametrized
path problem.  
\par
\vskip 3pt
Let $L=(V,E)$ be a network
in a form of a directed graph with $n$ vertices,
among them a distinguished
set $X=\{x_1,...,x_k\}$ of $k$ source vertices and a disjoint distinguished
set $Y=\{y_1,...,y_k\}$ of $k$ sink vertices.

A {\em walk} in $L$ is a sequence of vertices
$v_1,v_2,...,v_l$ of $L$
such that for $j=1,...,l-1,$ 
$( v_j,v_{j+1})\in E,$
$v_1$ is in $X,$ $v_2,...,v_{l-1}$ are in $V\setminus (X\cup Y),$
$v_l$ is in $Y.$ The length of the walk is $l-1.$
In other words, a walk is just a (not necessarily simple)
path starting from a vertex in $X$, having intermediate vertices
in $V\setminus (X\cup Y),$ and ending at a vertex
in $Y.$

A {\em proper set $S$ of walks} in $L$
is a set $W$ of $k$ walks 
of total length $\le k(n-1)$, each with a distinct start vertex in $X$
and a distinct end vertex in $Y.$
%each with a distinct
%start vertex in $X$ and a distinct end vertex in $Y.$
 
A {\em signature} of a proper set $S$ 
of walks is the pair $(i,j)$ that
is smallest in lexicographic
order such that the two walks that start at
$x_i$ and $x_j$ respectively intersect, and the first
intersection vertex of these two walks is the first
intersection vertex of the walk starting from $x_i$
with any walk in $S.$

\junk{
simple path $v_0,...,v_k,...,v_l$ such that 
$v_0,...,v_k$ and the reversal of $v_k,...,v_l$ are
prefixes of two different walks in $S.$}

Note that walks in $S$ are pairwise vertex disjoint iff
the signature  of $S$ is not defined.

We define the transformation $\phi$ on $S$ as follows.
%Order the walks in $S$ lexicographically.
If $S$ has the signature  $(i,j)$
then $\phi$ switches the suffix of the 
walk starting at $x_i$ with that of the walk
starting at $x_j$ at the first intersection vertex
of these two walks. See Fig. 1.
\junk{first pair of walks
in $S$ such that one of them
has prefix  $v_0,...,v_k$ and the other has as a prefix
the reverse of $v_k,...,v_l$, respectively.}
Otherwise, if the signature of $S$ is not
defined then $\phi$ is an identity on $S.$

\vspace{2mm}
\begin{figure}
 \label{fig:paths}
\begin{center}
\includegraphics[height=5cm]{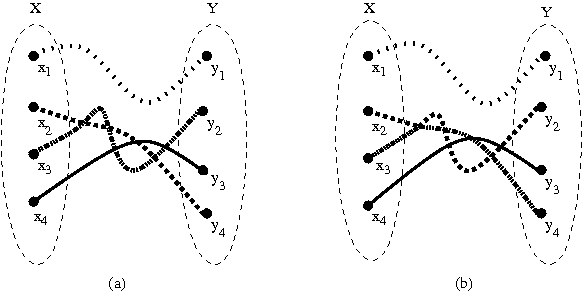}
%\centerline{\psfig{figure=parpaths.png,height=5cm}}
\end{center}
\caption{An example of a proper set $S$ of walks and
the companion proper set $\phi(S)$ of walks. }
\end{figure}
\vspace{2mm}

Observe that if the signature of $S$
is defined then $\phi(S)$ has the same signature
as $S$ and $\phi(S)\neq S.$ The first observation
is immediate. To show the second one it is sufficient
to note that $\phi(S)=S$ holds iff $\phi$  transforms the
two walks which yield the signature of $S$ onto themselves.
The latter is however impossible since they have different
start vertices and different end vertices. Note also
that the walks in $\phi(S)$ have the same total length as 
those in $S.$

It follows that $\phi$ is an involution on sets
of proper walks of total length $l,$ i.e.,
$\phi(\phi(S))=S$ holds for any  proper set $S$ of walks of 
total length $l$.

%For a finite set $X$ of variables, let $p_K(X)$
%be the polynomial which is the sum of $\choose |X| k$ distinct
% monomials which are products of $k$ distinct variables in $X.$

For the network $L$  and $l\in [k(n-1)],$
%with a distinguished set
%$X$ of sources and a disjoint, distinguished
%set of sinks $Y,$ where $|X|=|Y|,$
let $F_{L,l}$ be the family of all proper 
sets of $k$ walks of total length $\le l$ in $L.$
Assign a distinct variable $x_e$ to each edge
$e$ in $L.$ 
For a walk $W\in F_{L,l},$ let
$M_W$ be the monomial, where
$x_e$ has multiplicity equal
to the number of occurrences
of $e$ in $W.$
Next, let $Q_{L,l}$ denote the polynomial
$\sum_{S\in F_{L,l}}\prod_{W\in S}M_W.$

\begin{lemma} \label{lem: zero}
For the network $L$ and $l\in [k(n-1)],$ 
%with a distinguished set
%$X$ of sources and a disjoint, distinguished
%set of sinks $Y,$ where $|X|=|Y|,$
there is a proper set of 
$k$ mutually vertex-disjoint walks of total length $\le l$
in $L$ iff $Q_{L,l}$ is not identical
to zero over a field of characteristic two.
\end{lemma}

\begin{proof}
$F_{L,l}$ can be partitioned into the family $F^1_{L,l}$
of sets $S$ of walks such that $\phi(S)=S$ and 
the family $F^2_{L,l}$
of sets $S$ of walks such that $\phi(S)\neq S.$
The polynomial $\sum_{S\in F^2_{L,l}}\prod_{W\in S}M_W$ is identical to zero
over a field of characteristic two
since for each $S\in F^2_L$ the monomials
$\prod_{W\in S}M_W$ and $\prod_{W\in \phi(S)}M_W$
contain equal multiplicities of the same variables and
$\phi(\phi(S))=S$
so $S$ and $\phi(S)$ can be paired.
On the other hand, since each set $S$ of walks in $F^1_{L,l}$
consist of mutually vertex-disjoint walks, the monomials $\prod_{W\in S}M_W$
in the polynomial $\sum_{S\in F^1_{L,l}}\prod_{W\in S}M_W$ are 
in one-to one correspondence with $S$ and thus
are unique provided that the walks in $S$ are simple
paths. Now, it is sufficient to observe that mutually
vertex-disjoint  walks can be always trivially pruned
to corresponding mutually vertex-disjoint simple paths.

\qed \end{proof}

\junk{The next lemma provides an equivalent form of $Q_L$
that will be useful in its evaluation.

\begin{lemma} \label{lem: form}
For $l\in [n-1],$
$x\in X$ and $z\in Y,$ 
let $W_l(x,z)$ be the set of walks 
of length $l$ in $L$
that start at $x$ and end at $z.$
The polynomial $Q_L$ is equivalent to
$\sum_{\pi}\prod_{i=1}^k\sum_{l\in [n-1]}\prod_{W\in W_l(x_i,y_{\pi(i)}) }M_W.$
\end{lemma}}

To warm up, we prove the following lemma on sequential
evaluation of $Q_L.$

\begin{lemma}\label{lem: eval}
$Q_{L,l}$ can be evaluated for a given assignment
of values over a field $F_{2^{O(\log n)}}$ of characteristic two
in $O(k^3n^4+2^{k}k^4n^{3})$ time.
\end{lemma}
\begin{proof}
\junk{%Let $W_l(x,z)$ denote the set of walks
%from $x$ to $z$ of length $l$ in $L.$
We shall use the equivalent form
of $Q_L$ given in Lemma \ref{lem: form}.
We denote the value of a polynomial
$Q$ under the assignment $f$ by $f(Q).$ 

For $A\subset X$ and $B\subset Y,$ where
$|B|=|A|-1, l\in [(n-1)|B|],$ and $z\in V\setminus (X\cup Y),$
let  $W_l(A,B,z)$ be the family of all sets 
$S$ consisting of $|B|$ walks 
connecting the $|B|$ distinct sources in $A$
with the $|B|$ distinct sources in $B$ and a
not necessarily simple path starting
at the remaining source in $A$ and going
through some vertices in $V\setminus (X\cup Y)$,
and ending at $z$ so that the total length
of the walks and path is exactly $l.$
With each such a set $S$ we associate a 
monomial $M_S$ which is the product
of the monomials $M_W$ over the walks it
consists of and the multi-set
of the variables associated
with the edges on the path to $z$.
Let $P_l(A,B,z)$ be the polynomial
$\sum_{S\in W_l(A,B,z)}M_S.$}
For $B\subset Y,$ $l\in [(n-1)|B|],$ we consider
the family $W_l(B)$ of all sets 
$S$ consisting of $|B|$ walks 
connecting $|B|$ distinct sources in $\{x_1,...,x_{|B|}\}$
with the $|B|$ distinct sinks in $B$ so that the total length
of the walks is exactly $l.$ Next, we
define the polynomial $Q_p(B)$
as $\sum_{S\in W_l(B)}\prod_{W\in S}M_W.$ 
Note that $Q_{L,l}=\sum_{p=k}^lQ_p(Y)$
and $l\le nk.$

On the other hand, for $p\in [k(n-1)],$
$x\in X$ and $z\in V\setminus X,$ we consider
the set $W_p(x,z)$ of walks 
of length $p$ in $L$
that start at $x$ and end at $z.$
Let $Q_p(x,z)$ be the polynomial
$\sum_{W\in W_p(x,z)}M_W.$

We have the following recurrence for a nonempty subset $B$
of $Y$ 
%of the same cardinality as $B$ 
and $p\in [|B|(n-1)]:$

$$Q_p(B)=\sum_{y\in B}\sum_{q\in
  [|B|(n-1)-|B|+1]}Q_{p-q}(B\setminus \{ y\})Q_q(x_{|B|},y).$$

Next, we have also the following recurrence for $x\in X,$
$z\in V\setminus X,$ and $q\in [k(n-1)]:$

$$Q_q(x,z)= \sum_{u\in V\setminus (X\cup Y)\& (u,z)\in E}Q_{q-1}(x,u)x_{(u,z)}.$$ 

We have also $Q_1(x,z)=x_{(x,z)}$
if $(x,z)\in E,$ and otherwise $Q_1(x,z)=0.$
Consequently, we can evaluate
all the polynomials 
$Q_q(x,z)$ by the second recurrence in $O(k^2n^3)$ time.

Now, by using the first recurrence
and setting $Q_p(\emptyset)$ to $1$
in the field, we can evaluate all 
the polynomials
$Q_p(B)$ in the increasing order of the cardinalities
of $B$ in $O(2^{k} k^3n^2)$ time.

By $Q_{L,l}=\sum_{p=k}^lQ_p(Y)$ and
$l\le nk,$ we conclude that
$Q_{L,l}$ can be evaluated in $O(k^3n^4+2^{k}k^4n^{3})$ time.

\junk{Note that $Q_L=\sum_{l=k}^{k(n-1)} P_l(A,B).$
Let $\max (B)$ be the sink in $B$ with the largest
index. We assume $P_0(...)=0$ and 
if $(x,z)\in E$ then $P_1(\{x\},\emptyset)=x_{x,z}.$
Now, we can evaluate $P_l(A,B,z)$ and $P_l(A,B)$
by the following recurrences

$$f(P_{l+1}(A,B,z))= \sum {u\in V\setminus (X\cup Y)
\& (u,z)\in E} f(P_l(A,B,u))f(x_{(u,z)})+ \sum {y\in A\& (y,z)\in E}
 f(P_l(A\setminus \{y\},B))f(x_{(y,z)})$$

$$P_{l+1}(A,B)= \sum {z\in V\setminus (X\cup Y)
\& (u,z)\in E} P_l(A, B \setminus \{ \max (B))x_{(z,\max(B))}$$

in $O(2^{2k}kn^2)$ time.}
\qed \end{proof}

We can partially parallelize
the sequential evaluation of $Q_{L,l}$ in order to obtain
the following lemma.

\begin{lemma}\label{lem: pareval}
$Q_{L,l}$ can be evaluated for a given assignment 
of values over a field $F_{2^{O(\log n)}}$ of characteristic two
%to the variables
 in $O(k\log n +\log^2 n)$ time 
by a CREW PRAM using $O(k^2n^5 + 2^{k}k^4n^{3})$ processors.
\end{lemma}
\begin{proof}
We generalize the definition of the set $W_q(x,z)$
and the corresponding polynomial $Q_q(x,z)$
to include arbitrary start vertex $x\in V\setminus Y$,
requiring $z\in V\setminus X$ as previously.
Then, we can evaluate $Q_q(x,z)=\sum_{W\in W_l(x,z)} M_W$
for $x\in V\setminus Y$ and $z\in V\setminus X,$
for $q\in [k(n-1)]$ by the following standard
doubling recurrence for $q\ge 2:$

$$Q_q(x,z)=\sum_{y\in V\setminus (X\cup Y)}Q_{\lceil
  q/2\rceil}(x,y)Q_{\lfloor q/2\rfloor}(y,z)$$

At the bottom of the recursion, we have $Q_1(x,z)=x_{(x,z)}$
if $(x,z)\in E,$ otherwise $Q_1(x,z)=0.$
It follows that all $Q_q(x,z)$ 
for $x\in V\setminus Y$, $z\in V\setminus X,$
and $q\in[k(n-1)]$
can be evaluated
in a bottom-up manner in $O(\log^2 n)$ time 
by a CREW PRAM
using $O(kn^4)$ processors.

Recall the first recurrence from the proof of Lemma
\ref{lem: eval}.
When the polynomials $Q_q(x_i,y_j)$ for $x_i\in X$ and $y_i\in Y$
are evaluated, we can evaluate in turn 
the polynomials $Q_p(B)$,
where $B\subset Y,$ $p\in [|B|(n-1)]$
in $k$ phases in the increasing order of the cardinalities
of $B$  by this recurrence.
It can be done in $O(k(\log k +\log n))=O(k\log n) $ time
by a CREW PRAM using $2^{k}k^3n^2$ processors.

By $Q_{L,l}=\sum_{p=k}^lQ_p(Y),$
$l\le nk,$ and $k\le n,$ we conclude that
$Q_{L,l}$ can be evaluated in $O(k\log n +\log^2 n)$ time
by a CREW PRAM using $O(k^2n^5 + 2^{k}k^4n^{3})$
processors.

\junk{
Hence, the values $T(x_i,y_j)=\sum_{l\in \{ 1,...,n-1\}} T_l(x_i,y_j)$
for $i,j\in [k]$ can be also computed
in $O(\log^2 n)$ time 
using $O(n^3)$ processors.

Consequently,  by $f(Q_L)=\sum_{\pi}\prod_{i=1}^k T(x_i,y_{\pi(i)}),$ 
$f(Q_L)$ can be computed 
in $O(k\log k + \log^2 n)$ time 
by a CREW PRAM using $O(k!n)$ processors.}

\qed \end{proof} 
\junk{
\begin{lemma}\label{lem: eval}
$Q_L$ can be evaluated for a given assignment $f$
of values over a finite field of characteristic $2$
%to the variables 
in $O(k!k^3n^3)$ time.
\end{lemma}
\begin{proof}
%Let $W_l(x,z)$ denote the set of walks
%from $x$ to $z$ of length $l$ in $L.$
We shall use the equivalent form
of $Q_L$ given in Lemma \ref{lem: form}.
We denote the value of a polynomial
$Q$ under the assignment $f$ by $f(Q).$  .

To begin with, we compute
the values $T_l(x,z)=\sum_{W\in W_l(x,z)} f(M_W)$
for $x\in X$ and $z\in V\setminus X,$
for $l=1,...,n-1.$

\junk{It is easy to see that these values can be computed inductively
with respect to $l$ in $O(kn^2l)$ time. }

It is easy to see that these values can be computed
recursively by
$$T_l(x,z)= \sum_{y\in V\setminus (X\cup Y)} T_{l-1}(x.y)\sum_{W\in W_1(y,z)} f(M_W)$$ 
in $O(ln^2)$ time. 

By Lemma \ref{lem: form}, we have
$f(Q_L)=\sum_{\pi}\prod_{i=1}^k T(x_i,y_{\pi(i)}).$

\qed \end{proof}

We can also evaluate $Q_L$ efficiently in parallel.

\begin{lemma}\label{lem: pareval}
$Q_L$ can evaluated for a given assignment $f$
of values over a finite field of characteristic $2$
%to the variables
 in $O(k\log k + \log^2 n)$ time 
by a CREW PRAM using $O(k!n^3)$ processors.
\end{lemma}
\begin{proof}
We extend the definition of the set $W_l(x,z)$
to include arbitrary start vertex $x\in V\setminus Y$,
requiring $z\in V\setminus X$ as previously.
Then, we can evaluate the values $T_l(x,z)=\sum_{W\in W_l(x,z)} f(M_W)$
for $x\in V\setminus Y$ and $z\in V\setminus X,$
for $l=1,...,n-1$ by the following standard
doubling recurrence for $l\ge 2:$

$$T_l(x,z)=\sum_{y\in V\setminus (X\cup Y)}T_{\lceil
  l/2\rceil}(x,y)T_{\lfloor l/2\rfloor}(y,z)$$

It follows that all $T_l(x,z)$ 
for $x\in V\setminus Y$, $z\in V\setminus X,$
and $l=1,...,n-1$
can be computed
in a bottom up manner in $O(\log^2 n)$ time 
by a CREW PRAM
using $O(n^3)$ processors.

Hence, the values $T(x_i,y_j)=\sum_{l\in \{ 1,...,n-1\}} T_l(x_i,y_j)$
for $i,j\in [k]$ can be also computed
in $O(\log^2 n)$ time 
using $O(n^3)$ processors.

Consequently,  by $f(Q_L)=\sum_{\pi}\prod_{i=1}^k T(x_i,y_{\pi(i)}),$ 
$f(Q_L)$ can be computed 
in $O(k\log n + \log^2 n)$ time 
by a CREW PRAM using $O(k!n)$ processors.

\qed \end{proof} }

The following lemma on polynomial
identities verification has been shown independently by DeMillo and Lipton, 
Schwartz, and Zippel.

\begin{lemma}\label{lem: zip}\cite{DL78,S80}
Let $Q(x_1,x_2,...,x_m)$ be a nonzero polynomial
of degree $d$
over a field of size $r.$ Then, for $f_1,$ $f_2,$
...,$f_m$ chosen independently and uniformly
at random from the field, the probability that
$Q(f_1,f_2,...,f_m)$ is not equal to zero
is at least $1-\frac dr.$
\end{lemma}

Note that the polynomial $Q_{L,l}$ is of degree $l$
not larger than $k(n-1)\le n^2.$
We can use Lemma \ref{lem: zip} 
with a field $F_{2^{c\log n}}$ of characteristic two
to obtain a randomized test of the polynomial
$Q_{L,l}$ for not being identical to zero with
one side errors.
For sufficiently large constant $c,$
the one side errors are of probability 
not larger than a constant smaller than $1.$
By performing $O(n)$
such independent tests, the probability
of one side errors can be decreased
to exponentially
small in $n$ one. 

By Lemma \ref{lem: pareval}, the series of the tests can be performed 
in $O(k\log n + \log^2 n)$ time by a PRAM
using $O(k^2n^6 + 2^{k}k^4n^{4})$ processors.
\junk{in $n^{O(1)}$ time}
By Lemma \ref{lem: zero}, these tests
verify
if there is a proper set of mutually
vertex-disjoint walks 
of total length $\le l$ in the network $L$.
The latter in turn is equivalent to  
the existence of $k$ mutually vertex-disjoint paths
of total length $\le l$
connecting $X$ with $Y$
in $L$ by the definition
of a proper set of walks in $L.$
Hence, observing that each walk
can be trivially pruned to a simple directed
path with the same endpoints,
we obtain our main result.

\begin{theorem}\label{theo: pdec}
The problem of whether or not there
is a set of $k$ mutually vertex-disjoint 
simple directed paths of total length $\le l$
connecting $X$ with $Y$ in the network $L$
can be decided by a randomized CREW PRAM, with
one-sided errors of exponentially small probability in $n,$
running in $O(k\log n + \log^2 n)$
time and using $O(k^2n^6 + 2^{k}k^4n^{4})$ processors.
\end{theorem}

\junk{
\begin{proof}
To begin with, 
let us observe that it is
sufficient to ensure the relaxed thesis
with one sided errors of probability
not exceeding a constant $c<1$ in order
to obtain the original thesis by 
iterating the algorithm linear number of times.

Next, to prove the relaxed thesis
observe that there is a set of $k$ vertex disjoint 
simple paths $P_i$
connecting $x_i$ with $y_i$ for
$i=1,...,k$ in the network $L$
iff there is a mutually vertex-disjoint proper set
of walks in $L.$ 

Now, it is sufficient to combine Lemmata \ref{lem: eval}
with Lemma \ref{lem: zip}  applied to the field
$F_{2^{C\log n}}$ of characteristic two
for sufficiently large constant $C.$
\qed \end{proof} 
}

\section{Vertex-disjoint connecting paths of bounded cost}

In this section, we shall consider a more general
situation where there are a positive integer $C$
and  a cost function $c$
assigning to each of the $m$ edges $e$ in the network $L$ a cost
$c(e)\in [C].$ The cost of a walk or a path
is simply the sum of the costs of the edges
forming it (the cost of an edge is counted
the number of times it appears on the walk or path).
We would like to detect
a proper set of $k$ walks in $L$ that achieves
the minimum cost. 

For this reason, we consider the following
generalization of the polynomial $Q_{L,l}.$
For $U\in [mC],$
let $H_{L,U}$ be the set of all proper sets
of walks in the edge-costed network $L$
that have total cost not greater than $U.$
Next, for a walk $W$ in $L,$ 
as previously, let $M_W$
be the monomial which is the product
of $x_e$ over the occurrences
of edges $e$ on $W.$
The polynomial $CQ_{L,U}$ is defined
by $\sum_{S\in H_{L,U}}\prod_{W\in S} M_W.$

By using the proof method of 
Lemma \ref{lem: zero}, we obtain the following counterpart of this
lemma for $CQ_{L,U}$.

\begin{lemma} \label{lem: czero}
For the edge-costed network $L$,
%with a distinguished set
%$X$ of sources and a disjoint, distinguished
%set of sinks $Y,$ where $|X|=|Y|,$
there is a proper set of 
$k$ mutually vertex-disjoint walks
of total cost $\le U$ in $L$ 
iff $CQ_{L,U}$ is not identical
to zero over a field of characteristic two.
\end{lemma}

Next, we obtain
the following counterpart of Lemma \ref{lem: pareval} for $CQ_{L,U}$.

\begin{lemma}\label{lem: cpareval}
$CQ_{L,U}$ can be evaluated for a given assignment $f$
of values 
 over a field $F_{2^{O(\log n)}}$ of characteristic two
%to the variables
in $O(k\log (Cn) + \log^2 (Cn))$ time 
by a PRAM using 
$O(k^2C^5n^{10} + 2^{k}k^4C^3n^{6})$ processors.
\end{lemma}

\begin{proof} 
The proof reduces to that of Lemma \ref{lem: pareval}.
We replace each directed edge $e$ of cost $c(e)\in [C]$ 
in the network $L$  by
a directed path of length $c(e)$ introducing
$c(1)-1$ additional vertices. With each edge on such a
path, we associate a variable.
We assign $f(x_e)$ to the variable associated with the first
edge on the path replacing $e,$ 
and just $1$ of the field to the
variables associated with the remaining edges on the
path.

The resulting network $L'$ is of size $O(Cn^2).$
Let $H_{L',U}$ be the family of all
proper sets of $k$ walks of total cost $\le U$
in the network $L'$.
We can evaluate the polynomial
$Q_{L',U}=\sum_{S\in H_{L',U}}\prod_{W\in S}M_W$ 
in parallel analogously
as $Q_{L,l}$ in the proof of Lemma \ref{lem: pareval}.
It remains to observe that the value
of $CQ_{L,U}$ under the assignment $f$ is equal to that
of $Q_{L',U}$ under the aforementioned assignment.
\qed \end{proof}

Now, we are ready to derive our main result
in this section.

\begin{theorem}
The minimum cost of 
a set of $k$ mutually vertex-disjoint 
simple directed paths 
connecting $X$ with $Y$ in the 
network $L$ with edge costs in $[C]$
can be computed by a randomized CREW PRAM, with errors
of exponentially small probability in $n,$
running in
$O(k\log (Cn) + \log^2 (Cn))$ time 
and using 
$O(k^2C^6n^{13} + 2^{k}k^4C^4n^{9})$ processors. 

\end{theorem}

\begin{proof}
The minimum cost of the sought set of vertex-disjoint paths
is in $[Cn^2].$ Hence, by Lemma \ref{lem: czero}, 
it is sufficient to test the polynomials $CQ_{L,U}$
for non-identity with zero for all $U\in [Cn^2]$
in parallel. By applying Lemmata \ref{lem: zip} and \ref{lem: cpareval}
in a manner analogous to the proof of Theorem \ref{theo: pdec},
we conclude that 
it can be done by a randomized CREW PRAM, with
one-sided errors of exponentially small in $n$
probability, running in 
$O(k\log (Cn) + \log^2 (Cn))$ time 
and using\\ 
$O(Cn^2\times n \times (k^2C^5n^{10} + 2^{k}k^4C^3n^{6}))$
processors. 
\qed \end{proof}

\section{Finding vertex-disjoint connecting paths}

A straightforward approach of extending
our randomized parallel method for deciding if
there is a proper set of $k$ mutually vertex-disjoint
walks (of a bounded total cost)
between two sets of vertices of cardinality
$k$ to include the finding variant
could be roughly as follows. 
In parallel, for each $k$-tuple of respective neighbors
of the $k$ start vertices in $X,$
replace the set of start
vertices by the $k$-tuple
and apply our method recursively to the
resulting network. If the test is positive,
the first edges on the walks are known,
and we can iterate the method. 
The problem with this approach
is that its recursive depth is proportional
to the maximum length of a walk
in the resulting set of mutually vertex-disjoint
walks between $X$ and $Y.$
\junk{Thus, this approach could lead to a fast parallel
algorithm solely when there is a set 
of mutually vertex-disjoint
walks between the two sets satisfying the
total cost criterion, where each walk
has at most poly-logarithmic length.}

Also, it is not clear how one could
implement 
a straightforward divide-and-conquer approach
of guessing intermediate vertices
in order to find a set of
$k$ mutually-vertex disjoint walks 
of a given cost efficiently
in parallel. 
\junk{Namely, one could guess a  set of
intermediate vertices on the walks in
the aforementioned set, one for each walk,
form new terminal sets including copies
of the intermediate vertices and run
our test method on the resulting network.
The problem with this latter approach is that
the number of terminals, i.e., start and end vertices,
would grow exponentially in the recursion depth.
Consequently, the applications of our test method would require
too much time and too many processors.
Interestingly enough, at the bottom of 
the divide-and-conquer method, we would
encounter perfect bipartite matching problems
which admit RNC solutions not only
in the decision variant but also the finding
variant \cite{KUW86,MVV87,V93}.}
%, relying on the randomized test of a different
%polynomial, the Tutte's one.

We need more advanced methods
to obtain a very fast parallelization of the
finding variant. 
We shall modify the edge cost in
the network $L$ 
%an edge weighted version of the problem 
%of finding a set of $k$ mutually-vertex 
%disjoint walks between $X$ and $Y$
in order to use the so called {\em isolation lemma}
in a manner analogous to the RNC method of finding a
perfect matching given in \cite{V93}. 

\begin{lemma}(The isolation lemma \cite{V93}).
Let $F$ be a family of subsets of a set with $q$ elements
and let $r$ be a non-negative integer. 
Suppose that each element $s$ of the set is independently 
assigned a weight $w(s)$ uniformly at random from 
$[r]$, and the weight of a subset $S$ in $F$ is defined as
$w(S) = \sum_{x \in S} w(x).$
Then, the probability that there is a unique set in 
$F$ of minimum weight is at least $1 -\frac qr.$
\end{lemma}

\begin{corollary}\label{cor: ciso}
For each of the $m$
edges $e$ in the network $L,$ 
modify its cost $c(e)$ to $c'(e)=c(e)rm +w(e),$
where the weight $w(e)$ is drawn uniformly at random from 
$[r]$.
%, and let the weight of a 
%set of $k$ mutually vertex-disjoint paths
%connecting $X$ with $Y$ be the sum
%of the weights of the edges on the paths
%in this set. 
Then, the probability that there 
is a unique minimum-cost set of mutually vertex-disjoint paths
connecting $X$ with $Y$ in the
edge weighted network $L$ is at least $1 -\frac mr.$
\end{corollary}
\begin{proof}
To use the isolation lemma, let 
the underlying set to consist
of all edges in the network $L.$ Next, note
that a set of mutually vertex-disjoint paths
connecting $X$ with $Y$ achieving a minimum cost
consists of simple paths and thus it can be identified
with the set of edges on the paths.
Let $P$ be the family of all sets of
mutually vertex-disjoint simple paths
connecting $X$ with $Y$ in the network $L.$
By the setting of new costs $c'(e),$ solely
those sets in $P$ that achieved the minimum cost, say $D,$
under the original costs $c(e)$ can achieve
a minimum cost under the new costs $c'(e).$
So, we can set $F$ to the aforementioned sub-family of $P$,
and define the weight of a 
set of $k$ paths in $F$ as the sum
of the weights $w(e)$ of the edges $e$ on the paths
in this set in order to use the isolation lemma. 
By the isolation lemma, there is a unique set $S$ in $F$
that achieves the minimum weight $w(S)$ with the
probability at least $1 -\frac mr.$ The corollary
follows since each set $S$ in $F$ has the
cost $c'(S)$ equal to $D+w(S).$
\qed \end{proof}

Throughout the rest of this section, we shall assume 
that each of the $m$
edges $e$ in the network $L$ is assigned the cost $c'(e)$
as in Corollary \ref{cor: ciso} and that $r\in [n^{O(1)}].$ 

Suppose that we know the minimum cost of a set of $k$
mutually vertex-disjoint paths connecting $X$ with $Y$
in the network $L$ with the edge costs indicated, and such
a minimum-cost set is unique. Then, it is sufficient
to show that we can test quickly in parallel if the 
network $L$ with an arbitrary edge removed
still contains a set of $k$
mutually vertex-disjoint paths connecting $X$ with $Y$
that achieves the minimum cost. By performing the
test for each edge of $L$ in parallel, we
can determine the set of edges forming the
unique minimum-cost set of $k$
mutually vertex-disjoint paths connecting $X$ with $Y$.

To carry out these tests, we need to
generalize the polynomial $CQ_{L,U}$ to a polynomial $CP_{L,e,U}$,
where $e$ is an edge in $L$ and
$U$ is a cost constraint from $[mr(mC+1)]=[Cn^{O(1)}]$.
Let $H_{L,e,U}$ be the family of all proper 
sets of $k$ walks in the network $L$ 
with the edge $e$ removed that
have total at most $U.$
(In the total cost of a set of walks, 
we count the cost of an edge  the number 
of times equal to the sum of the multiplicities of the edge
in the walks.)

As in the definition
of $Q_{L,l}$
assign a distinct variable $x_e$ to each edge
$e$ in $L,$ and for a walk $W\in H_{L,e,W},$ let
$M_W$ be the monomial, where
$x_e$ has multiplicity equal
to the number of occurrences
of $e$ in $W.$ The polynomial
$CP_{L,e,U}$ is defined by
$\sum_{S\in H_{L,e,W}}\prod_{W\in S}M_W.$

By using the proof method of 
Lemma \ref{lem: zero}, we obtain the following counterpart of this
lemma for $CP_{L,e,U}$.

\begin{lemma} \label{lem: gzero}
For the edge-costed network $L$ with $m$ edges, 
edge $e,$ and $U\in [Cn^{O(1)}],$ 
%with a distinguished set
%$X$ of sources and a disjoint, distinguished
%set of sinks $Y,$ where $|X|=|Y|,$
there is a proper set of 
$k$ mutually vertex-disjoint walks
of total cost $\le U$ in 
the network $L$ with the edge $e$ removed
iff $CP_{L,e,U}$ is not identical
to zero over a field of characteristic two.
\end{lemma}

Next, we obtain
the counterpart of Lemma \ref{lem: pareval} for
$CP_{L,e,U}$ following the proof of Lemma \ref{lem: cpareval}.

\begin{lemma}\label{lem: gcpareval}
$CP_{L,e,U}$ can be evaluated for a given assignment 
of values 
 over a field $F_{2^{O(\log n)}}$ of characteristic two
%to the variables
in $O(k\log (Cn) + \log^2 (Cn))$ time 
by a PRAM using 
$2^{k}(kCn)^{O(1)}$ processors.
\end{lemma}

Now, we are ready to derive our main result in this section.

\begin{theorem}\label{theo: cpfind}
There is a randomized PRAM returning 
almost certainly 
(i.e., with probability
at least $1-\frac 1{n^{\alpha}}$, where $\alpha \ge 1$)
a minimum-cost set of $k$ mutually vertex-disjoint paths 
connecting
$X$ with $Y$ in the network $L$ 
with the original edge costs in $[C]$
(iff such a set exists)
in $O(k\log (Cn) + \log^2 (Cn))$ time  using 
$2^{k}(kCn)^{O(1)}$ processors.
\end{theorem} 
\begin{proof}
We set $r$ to, say, $n^2m,$ and specify
the new edge costs $c'(e)$ in the network $L$ 
drawing the weights $w(e)$ uniformly at random
from $[r]$
as in Corollary \ref{cor: ciso}.
Next, for each $U \in [mr(mC+1)]=[Cn^{O(1)}],$ we proceed in
parallel as follows. For each edge $e$ of the network
$L$, we test the polynomial $CP_{L,e,U}$ for the non-identity with
zero by using Lemma \ref{lem: zip}
and Lemma \ref{lem: gcpareval} (we can perform a linear in
$n$ number of such tests in parallel in order to decrease
the probability of the one-sided error to an exponentially small
one). Next, we verify if the edges that passed the test
positively yield a set of $k$ mutually vertex-disjoint paths 
connecting
$X$ with $Y$. For example, it can be done
by checking for each endpoint of the edges outside $X\cup Y$ 
if it is shared by exactly two of the edges, and then computing and examining
the transitive closure of the graph induced by the edges
(see \cite{R93}).
If so, we save the resulting set of paths
of total (new) cost $\le U.$
By Corollary \ref{cor: ciso}, there is a $U \in [Cn^{O(1)}]$ for which
the above procedure will find
such a set of paths that achieves the minimum (original) cost with
probability at least $1-\frac 1 n.$
\qed \end{proof}

\section{Minimum-cost logarithmic integral flow is in $RNC^2$}

The following lemma is a straightforward
generalization of a folklore reduction
of maximum integral flow to a corresponding
disjoint connecting path problem (for instance cf. \cite{E79})
to include minimum-cost integral flow. We shall call a flow proper,
if it ships each flow unit along
a simple path from the source to the sink.

\begin{lemma}\label{lem: red}
The problem of whether or not 
there is a proper integral flow 
of value $k$ and cost $D$ from
a distinguished source vertex
$s$ to a distinguished sink vertex $t$
in a directed network with $n$ vertices.
integral
edge capacities and edge costs in $[C]$
%that are polynomially bounded
%with respect to $n$
%is log space as well as $NC^2$ reducible 
can be (many-one) reduced
to that
of whether or not there is
set of $k$ mutually vertex-disjoint 
simple directed paths of total cost $D^*$,
where $\lfloor \frac {D^*} {kn} \rfloor =D,$
connecting two distinguished sets
of $k$ vertices in a directed
network on $O(kn^2)$ vertices
%in log space as well as 
in $O(\log k +\log n)$
time by a CREW PRAM
using $O(kn^2+k^2n)$ processors.
\end{lemma}
\begin{proof}
Let $K=(V,E)$ be the directed network
with integral edge capacities, edges costs
in $[C]$ and 
the distinguished source vertex
$s$ and sink vertex $t.$ 
Since we are interested in a flow 
of value $k$, we can assume w.l.o.g that
all edge capacities do not exceed $k.$

We form a directed network $K^*$ 
on the basis of the network $K$ as follows.

Let $v\in V.$ Next, let $E_{in}(v)$ be the set of edges
in $K$ incoming into $v$, and let $E_{out}(v)$ be
the set of edges in $K$ leaving $v.$ 
For each $e\in E_{in}(v)$ and $i\in [capacity(e)],$
we create the vertex $v_{in}(e,i).$ Analogously,
for each  $e'\in E_{out}(v)$ and
$i'\in [capacity(e')],$
we create the vertex $v_{out}(e',i').$
Furthermore, we direct an edge from each vertex $v_{in} (e,i)$
to each vertex $v_{out} (e',i')$. To each such an edge,
we assign the cost $1.$
Also, for each edge $e=(v,w)$ of $K,$ 
we direct an edge from $v_{out} (e,i)$ to 
$w_{in}(e,i)$ for $i\in [capacity(e)].$ To each such an edge,
we assign the cost $c(e)kn.$
See Fig. 2.
\begin{figure}
 \label{fig:paths1}
\begin{center}
\includegraphics[height=5cm]{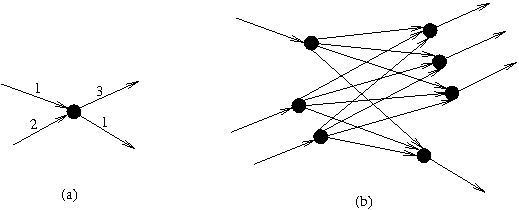}
%\centerline{\psfig{figure=red.png,height=5cm}}
\end{center}
\caption{An example of a vertex of 
the network $K$ and the corresponding
part of the network $K^*.$}
\end{figure}
Let $X'$ be the set of vertices of the form $s_{out}(...)$,
and let $Y'$ denote the set of vertices of the form
$t_{in}(...).$ Create an additional set $X$ of $k$
vertices and from each vertex in $X$ direct an edge
to each vertex in $X'$. Symmetrically, create another additional
set $Y$ of $k$ vertices and from each vertex in $Y'$ direct an
edge to each vertex in $Y.$

It is easy to observe that 
there is a proper integral flow 
of value $k$ and cost $D$  from
$s$ to $t$ in the network $K$ iff there is
a set of $k$ mutually vertex-disjoint 
simple paths of total cost $D^*$
connecting $X$ with $Y$ in the network $K^*,$
such that $\lfloor \frac {D^*} {kn} \rfloor =D.$

Now it is sufficient to note that the construction
of $K^*$, $X$ and $Y$  on the basis of $K$ 
%can be done
%in space logarithmic in the size of the network $K$ and $k.$
%It can be also 
easily implemented by a CREW PRAM
in $O(\log k +\log n)$-time using 
$O(kn^2+k^2n)$ processors,
where $n$ is the number of vertices in $K.$
%which yields an $NC^2$ reduction.
\qed \end{proof}

By combining Theorem \ref{theo: pdec} with Lemma \ref{lem: red}, we obtain
our first main result.

\begin{theorem}
The minimum cost of
a flow of value $k$ in a network
with $n$ vertices, a sink and a source, 
integral edge capacities and positive integral
edge costs in $[C]$
can be found by a randomized PRAM, with errors
of exponentially small probability in $n,$
running in $O(k\log (Ckn)+\log^2 (Ckn)$ time
and using 
%$O(k^{13}C^6n^{26} + 2^{k}k^{13}C^4n^{18})$
$2^{k}(kCn)^{O(1)}$ processors. 
%$O(k^9C^4n^{16} + 2^{k}k^{8}C^2n^{10})$.
\end{theorem}

\junk{Note that $\log (ab) \le \log a \log b$ for $a,b\ge 10.$

\begin{corollary}
The problem of detecting the minimum
cost of a flow of value $k$
in a network
with $n$ vertices, a sink and a source, 
integral edge capacities and positive integral
edge costs polynomially bounded in $n$
admits a randomized FPP algorithm.
\end{corollary}}

By combining in turn
Theorem \ref{theo: cpfind} with 
the finding variant of Lemma \ref{lem: red}
using exactly the same reduction, we obtain
our second main result.

\begin{theorem}\label{theo: ffind}
There is a randomized PRAM algorithm returning 
 almost certainly 
%(i.e., with probability
%at least $1-\frac 1{n^{\alpha}}$, where $\alpha \ge 1$)
a minimum-cost flow of value $k$ (iff a flow of value
$k$ exists) in a network 
with $n$ vertices, a sink and a source, 
integral edge capacities and edge costs in $[n^{O(1)}],$
in $O(k\log (kn)+\log^2 (kn))$ time using
$2^{k}(kn)^{O(1)}$ processors.
\end{theorem}

\junk{
\begin{corollary}
The problem of finding a minimum-cost flow of value $k$ in a network
with a sink and a source, integral edge capacities 
and positive integral edge costs polynomially
bounded in the size of the network
admits a randomized FPP algorithm.
\end{corollary}}

\begin{corollary}
The problem of finding a 
minimum-cost flow of value $O(\log n)$ in a network
with $n$ vertices, a sink and a source, and integral edge capacities 
bounded polynomially in $n$  admits an $RNC^2$ algorithm.
\end{corollary}

\section{Final remarks}

We have resented a new approach 
to the minimum-cost integral
flow problem. In particular, it yields an $RNC^2$ algorithm
when the flow supply is (at most)
logarithmic in the size of the network.
 
All our results can be extended to
include undirected networks by
a straightforward reduction.

\junk{The question arises
if one can obtain a deterministic FPP method 
for the problem of finding  
a minimum-cost  integral flow of value $k$
from the source $s$ to the sink $t$ ?

A straightforward approach of extending
our randomized method for deciding if
there is a proper set of mutually vertex-disjoint
walks between two sets of vertices of cardinality
$k$ could be roughly as follows. 
In parallel, for each $k$-tuple of respective neighbors
of the $k$ start vertices in $X,$
replace the set of start
vertices by the $k$-tuple
and apply our method recursively to the
resulting network. If the test is positive,
the first edges on the walks are known,
and we can iterate the method. 
The problem with this approach
is that its recursive depth is proportional
to the maximum length of a walk
in the resulting set of mutually vertex-disjoint
walks between $X$ and $Y.$
\junk{Thus, this approach could lead to a fast parallel
algorithm solely when there is a set 
of mutually vertex-disjoint
walks between the two sets, where each walk
has at most poly-logarithmic length.}

One could also consider
a straightforward divide-and-conquer approach
in order to find a set of
$k$ mutually-vertex disjoint walks fast
in parallel. Namely, one could guess a  set of
intermediate vertices on the walks in
the aforementioned set, one for each walk,
form new terminal sets including copies
of the intermediate vertices and run
our test method on the resulting network.
The problem with this approach is that
the number of terminals
\junk{, i.e., start and end vertices,}
would grow exponentially in the recursion depth.
\junk{Consequently, the applications of this test method would require
too much time and too many processors.}
Interestingly enough, at the bottom of 
the divide-and-conquer method, we would
encounter perfect bipartite matching problems
which admit RNC solutions not only
in the decision variant but also the finding
variant \cite{V93}.
%, relying on the randomized test of a different
%polynomial, the Tutte's one.

It seems that more advanced methods are needed
to obtain an efficient parallelization of the
finding variant. In fact, one could consider
an edge weighted version of the problem 
of finding a set of $k$ mutually-vertex 
disjoint walks between $X$ and $Y$
in order to use the so called isolation lemma
in a manner analogous to the RNC method of finding a
perfect matching from \cite{MVV87,V93}. 
By the isolation lemma 
if the edge weights were picked uniformly at random
from sufficiently large 
polynomial range then
there would be likely a unique minimum weight set of
$k$ mutually-vertex 
disjoint walks between $X$ and $Y$.
Unfortunately, it seems that just finding
the minimum weight of a set $k$ mutually-vertex 
disjoint walks between $X$ and $Y$ for
a random edge weight assignment in the range 
specified by the isolation lemma might
require $n^{\Omega (k)}$ work. Simply, potentially
there are so many ways of splitting the minimum weight
between the sought $k$-walks.

On the other hand, it seems that just showing
that the slices of the finding variants of our parametrized
network flow problem and the corresponding
$k$ path problem admit NC or RNC algorithms is of interest.
For instance, problems hard for the sequential parametrized
class FPT must have some slices not in NC if
NP is not equal to NC \cite{CD97}. 

Assume that the integral weights 
assigned to the edges of the network 
according to the isolation lemma resulted in  a unique
minimum weight set of mutually vertex-disjoint paths
connecting $X$ with $Y,$ and that we can compute
this minimum weight in NC.
Now, we can simply test each
edge $e$ of the network for the containment in the optimal
set of paths separately in parallel by verifying
if there is a set of mutually vertex-disjoint
paths connecting $X$ with $Y$ in the 
edge weighted network with
the edge removed $e$ whose total weight 
does not exceed the minimum weight. If so, then
we can account the edge to the sought set of
edges forming the unique
minimum weight set of mutually vertex-disjoint paths
connecting $X$ with $Y,$ Otherwise, we can discard it.
To perform the aforementioned verifications in RNC 
as well as to compute the minimum weight of a
set of mutually vertex-disjoint paths
connecting $X$ with $Y$ in $RNC$, 
it is sufficient to show
that the test if there is a set of mutually vertex-disjoint
paths of total weight not exceeding a polynomially
bounded value in the edge weighted network
is in RNC. This can be done by relatively straightforward
generalization of Theorem . We generalize the length
constraint $l$ from the proof of Theorem
to an integral weight constraint and replace the path doubling
recursion in Lemma with a bit more complicated one,
since we cannot split individual edge weights, one has try
more possible partitions of total path weight.
The details are left to the reader.

In fact, if we allow for the use of $n^{O(k)}$ processors
then we can solve the finding variant 
of the parametrized vertex-disjoint path problem
and the parametrized integral network flow
problem by receptively
constructing $k$ vertex cuts 
deterministically in poly-logarithmic time.} 

%===============================================================

\newcommand{\CIAC}{Italian Conference on Algorithms and Complexity}
\newcommand{\COCOON}{Annual International Computing Combinatorics Conference (COCOON)}
\newcommand{\COMPGEOM}{Annual ACM Symposium on Computational Geometry (SoCG)}
\newcommand{\ESA}{Annual European Symposium on Algorithms (ESA)}
\newcommand{\FOCS}{IEEE Symposium on Foundations of Computer Science (FOCS)}
\newcommand{\FSTTCS}{Foundations of Software Technology and Theoretical Computer Science (FSTTCS)}
\newcommand{\ICALP}{Annual International Colloquium on Automata, Languages and Programming (ICALP)}
\newcommand{\IPCO}{International Integer Programming and Combinatorial Optimization Conference (IPCO)}
\newcommand{\ISAAC}{International Symposium on Algorithms and Computation (ISAAC)}
\newcommand{\ISTCS}{Israel Symposium on Theory of Computing and Systems}
\newcommand{\JACM}{Journal of the ACM}
\newcommand{\LNCS}{Lecture Notes in Computer Science}
\newcommand{\MOR}{Mathematics of Operations Research}
\newcommand{\SICOMP}{SIAM Journal on Computing}
\newcommand{\SIJDM}{SIAM Journal on Discrete Mathematics}
\newcommand{\SODA}{Annual ACM-SIAM Symposium on Discrete Algorithms (SODA)}
\newcommand{\SPAA}{Annual ACM Symposium on Parallel Algorithms and Architectures (SPAA)}
\newcommand{\STACS}{Annual Symposium on Theoretical Aspects of Computer Science (STACS)}
\newcommand{\STOC}{Annual ACM Symposium on Theory of Computing (STOC)}
\newcommand{\SWAT}{Scandinavian Workshop on Algorithm Theory (SWAT)}
\newcommand{\TCS}{Theoretical Computer Science}

\newcommand{\Proc}{Proceedings of the }
\renewcommand{\Proc}{Proc. }

%===============================================================
\end{document}